\documentclass[a4paper,12pt]{article}
    \usepackage[top=2cm,bottom=2cm,left=2.5cm,right=2.5cm]{geometry}
    \usepackage{cite, amsmath, amssymb}
    \usepackage[margin=0.5cm,%
                font=small,%
                format=hang,%
                labelsep=period,%
                labelfont=bf]{caption}
\usepackage{amsmath}
\usepackage{amsfonts}
\usepackage{amssymb}
\usepackage{cite}
\usepackage{amsthm}
\usepackage{makeidx}
\usepackage{graphicx}
\usepackage{mathrsfs,xcolor,multicol}
\usepackage{enumerate}
\usepackage{blkarray}
\usepackage{bm}
\usepackage{setspace}
\usepackage{float}
\usepackage{authblk}
\usepackage{multirow}
\usepackage{booktabs}
\usepackage[ruled,vlined]{algorithm2e}
\usepackage{blkarray}
\usepackage{authblk}
\usepackage{multirow}
\usepackage{lineno}

\usepackage{blkarray, bigstrut}
\usepackage{authblk}
\usepackage{multirow}

\usepackage{fancyhdr}
\usepackage{lastpage}

\usepackage{amsmath, amssymb, booktabs, lastpage,}
\usepackage{multicol}
\usepackage{hyperref}
\usepackage{blkarray}
\usepackage{empheq} 

\numberwithin{table}{section}
\numberwithin{equation}{section}

\theoremstyle{plain}
\newtheorem{theorem}{Theorem}[section]

\newtheorem{definition}[theorem]{Definition}

\newtheorem{example}[theorem]{Example}

\newtheorem{remark}[theorem]{Remark}

\usepackage{authblk}
\author[1]{ \textbf{Bryan S. Hernandez}}
\author[2]{ \textbf{Patrick Vincent N. Lubenia}}
\author[2,3,4]{\textbf{Eduardo R. Mendoza}}

\affil[1]{\small \textit{Institute of Mathematics, University of the Philippines Diliman, Quezon City 1101, Philippines}}
\affil[2]{\small \textit{Systems and Computational Biology Research Unit, Center for Natural Sciences and Environmental Research, Manila 0922, Philippines}}
\affil[3]{\small \textit{Mathematics and Statistics Department, De La Salle University, Manila  0922, Philippines}}
\affil[4]{\small \textit{Max Planck Institute of Biochemistry, Martinsried near Munich 82152, Germany}}
\affil[*]{Email addresses: \texttt{bshernandez@up.edu.ph},
\texttt{pnlubenia@upd.edu.ph}, \texttt{eduardo.mendoza@dlsu.edu.ph}}

\title{\textbf{Embedding-based comparison of reaction networks of Wnt signaling}}
\date{}

\begin{document}
\maketitle
\begin{abstract} 
This work introduces a new method for comparing two reaction networks of the same or closely related systems through their embedded networks in terms of the shared set of species.  Hence, we call this method the Common Species Embedded Networks (CSEN) analysis. Using this approach, we conduct a comparison of existing reaction networks associated with Wnt signaling models (Lee, Schmitz, MacLean, and Feinberg) that we have identified.
The analysis yields three important results for these Wnt models. First, the CSEN analysis of the Lee (mono-stationary) and Feinberg (multi-stationary) shows a strong similarity, justifying the study of the Feinberg model, which was a modified Lee model constructed to study an important network property called ``concordance''.
It also challenge the absoluteness of discrimination of the models into mono-stationarity versus multi-stationarity, which is a main result of Maclean et al. (PNAS USA 2015).
Second, the CSEN analysis provides evidence supporting a strong similarity between the Schmitz and MacLean models, as indicated by the ``proximate equivalence'' that we have identified. Third, the analysis underscores the absence of a comparable relationship between the Feinberg and MacLean models, highlighting distinctive differences between the two. Thus, our approach could be a useful tool to compare mathematical models of the same or closely related systems.

\end{abstract}

\thispagestyle{empty}
\section{Introduction}
\label{sec:1}

Reaction networks have recently been used to perform parameter-free analyses of mathematical models individually and comparatively \cite{NFM2023,LUML2023,LUML2021}. The structure of the network's independent decompositions \cite{FeinbergBook,had2022,hernandez:delacruz1}, equilibria parametrizations \cite{NFM2019,HernandezetalPCOMP2023}, and a standard toolbox for reaction networks \cite{FeinbergToolbox} are some of the important tools that are used to analyze the models.
The crucial properties studied include multi-stationarity, absolute concentration robustness (ACR), and concordance. Multi-stationarity is the capacity for a network to admit two or more positive equilibria
with the same conserved quantities
\cite{dickenstein2019multistationarity}. On the other hand, a system exhibits ACR in a particular species if the equilibrium value for a particular species has the same value at any equilibrium regardless of any set of initial conditions \cite{SHFE2010}. Finally, concordance property is a structural feature of a network that enforce duller and more restrictive behavior despite what might be great intricacy in the interplay of many species \cite{SHFE2012}.

Embedded networks of reaction networks were first introduced in 2013 by Joshi and Shiu in \cite{JOSH2013} and were used to study continuous flow stirred tank reactors (CFSTR). Embedded networks are also the basis for embedded representations of Biochemical Systems Theory (BST) systems \cite{AJLM2017}. In particular, Farinas et al. \cite{FAML2021} used embedded representations to study a model of Mycobacterium tuberculosis growth. More recently, Meshkat et al. \cite{MEST2022} applied the concept to derive a necessary and sufficient condition to determine whether a reaction network has ACR property for a class of mass action systems.

In this work, we analyze the structural and kinetic relationships of reaction networks through their embedded networks relative to their set of common species. This was motivated by our analysis of Wnt signaling models in biochemistry. Wnt signaling is a vital mechanism that regulates crucial development processes and maintenance of tissue homeostasis.

Wnt signaling pathway have been analyzed via reaction networks by MacLean et al. \cite{Maclean} primarily to investigate bistability, i.e., existence of two stable positive equilibria for a particular set of rate constants. They compared the previous models of Lee \cite{Lee} and Schmitz \cite{Schmitz} with their proposed model \cite{Maclean}, which we call MacLean model in this paper. The species and reaction networks of the models are provided in the Appendices \ref{app:var} and \ref{reaction:networks}, respectively.

The Lee model focuses on the formation of the destruction complex from its constituent parts and how its subsequent ability to degrade $\beta$-catenin is altered by the presence and absence of an external Wnt stimulus. The model assumes that all species are uniformly distributed throughout the cell and hence does not distinguish between the nucleus and the cytoplasm.
The Schmitz model focuses on the effect that shuttling of $\beta$-catenin and destruction complex between the cytoplasm and the nucleus has on T-cell factor binding to $\beta$-catenin in the nucleus.
The MacLean model focuses on both $\beta$-catenin degradation and protein shuttling between the cytoplasm and nucleus that may serve as a possible mechanism for governing bistability in the pathway \cite{Maclean}.

{Feinberg \cite{FeinbergBook}  modified the Lee model by considering a complex of three species into a single species, making the reaction containing this complex a reversible one, and eliminating a single species. In this paper, we call this network the Feinberg network. When the pair of reversible reaction that does not contribute much to the network except to produce
a species that does not play any other role in the Wnt signaling network is removed, the concordance property is achieved from an originally discordant network.
}

The analysis yields three important results for the Wnt signaling models. First, the CSEN analysis of the Lee, which is mono-stationary, and Feinberg, which is multi-stationary, shows a strong similarity. This justifies the study of the Feinberg model, which was a slightly modified Lee model.
Furthermore, it challenges the absoluteness of discrimination of the models into mono-stationarity versus multi-stationarity.
Second, the CSEN analysis suggests a strong similarity between the Schmitz and MacLean models. This is indicated by the ``proximate equivalence'' that we have identified. Third, the analysis underscores the absence of a comparable relationship between the Feinberg and MacLean models. This highlights the distinctive differences between these two models.

The rest of the paper is organized as follows: Section \ref{chap:preliminaries} provides the background necessary to understand the flow of this paper.
Section \ref{CSEN:method} outlines the new method of comparing reaction networks via common species embedded networks.
Sections \ref{sec:comparisonLeeFAL}, \ref{section:Schmitz:MacLean}, and \ref{section:comp:FAL:MacLean} provide our CSEN analyses for the Lee and Feinberg models, the Schmitz and MacLean models, and the Feinberg and MacLean models, respectively. Finally, a summary and recommendation are given in Section \ref{sec:summary}.

\section{Preliminaries}
\label{chap:preliminaries}


\subsection{Chemical reaction networks and systems}

We start by introducing a formal definition of a chemical reaction network (CRN).

\begin{definition}
    A \emph{chemical reaction network}, denoted by $\mathscr{N}$, 
    is a triple of nonempty and finite sets $(\mathscr{S}, \mathscr{C}, \mathscr{R})$ where
    \begin{itemize}
        \item[a.] $\mathscr{S} = \left\{A_1, A_2, \ldots, A_m \right\}$ is the set of $m$ \emph{species},
        \item[b.] $\mathscr{C} = \{C_1, C_2, \ldots, C_n\}$ is the set of $n$ \emph{complexes}, which are non-negative linear combinations of the species, and
        \item[c.] $\mathscr{R} = \{R_1, R_2, \ldots, R_r\} \subset \mathscr{C} \times \mathscr{C}$ is the set of $r$ \emph{reactions}.
    \end{itemize}
    We usually denote a reaction $(y,y')$ using $y \to y'$. In this reaction, $y$ is called a \emph{reactant complex} and $y'$ is called a \emph{product complex}. {An inflow reaction, denoted by $0 \to A_i$, means a constant supply of species $A_i$. On the other hand, an outflow reaction, denoted by $A_j \to 0$, means degradation of species $A_j$.}
    
    A \emph{reaction vector} of the reaction is the difference $y'-y$. Hence, a reaction vector is a linear combination of species.
    The linear subspace of $\mathbb{R}^m$ spanned by the reaction vectors is called the \emph{stoichiometric subspace} of $\mathscr{N}$, defined as $S = \mathrm{span}\{y' - y \in \mathbb{R}^m \mid y \rightarrow y' \in \mathscr{R}\}$.
    The $m \times r$ matrix, where the $i$th column contains the coefficients of the associated species in the $i$th reaction vector associated with the $i$th reaction, is the \emph{stoichiometric matrix} of the network.
\end{definition}

To describe the dynamics of the temporal evolution of the concentrations of the species, a CRN is endowed with kinetics. Kinetics is defined as follows.

\begin{definition}
A \emph{kinetics} for a reaction network $\mathscr{N}=(\mathscr{S}, \mathscr{C}, \mathscr{R})$ is an assignment to
each reaction $y \to y' \in \mathscr{R}$ of a continuously differentiable {rate function} $K_{y\to y'}: \mathbb{R}^\mathscr{S}_{{\geq} 0} \to \mathbb{R}_{\ge 0}$ such that the following positivity condition holds:
$K_{y\to y'}(c) > 0$ if and only if ${\rm{supp \ }} y \subset {\rm{supp \ }} c$, where ${\rm{supp \ }} y$ refers to the support of the vector $y$, i.e., the set of species with nonzero coefficient in $y$.
The pair $\left(\mathscr{N},K\right)$ is called a \emph{chemical kinetic system}.
\end{definition}

\begin{definition}
A kinetics for a CRN $(\mathscr{S},\mathscr{C},\mathscr{R})$ is \emph{mass action} if for each reaction $y\to y'$ (i.e., $[y_1,y_2,\ldots, y_m]^\top \to [y_1',y_2',\ldots, y_m']^\top$),
$$K_{y\to y'}(x)=k_{y\to y'}\prod _{i \in \mathscr{S}} x_i^{y_i}$$
for some $k_{y\to y'}>0$.
\end{definition}

\begin{definition}
	The \emph{species formation rate function} (SFRF) of a chemical reaction system $(\mathscr{N},K)$ is given by $f\left( x \right) = \displaystyle \sum\limits_{{y} \to {y'} \in \mathscr{R}} {{K_{{y} \to {y'}}}\left( x \right)\left( {{y'} - {y}} \right)}.$
\end{definition}
Note that the SFRF can be written as
$f(x) = NK(x)$ where $N$ is the stoichiometric matrix of the network and $K$ is the vector of rate functions.
The system of \emph{ordinary differential equations} (ODEs) of a chemical kinetic system is given by $\dfrac{{dx}}{{dt}} = f\left( x \right)$ where $x$ is a vector of concentrations of the species 
that evolve over time.

\begin{definition}
A \emph{steady state} or an \emph{equilibrium} is a vector $c$ of species concentrations such that $f(c)=0$. A \emph{positive equilibrium} is an equilibrium where each concentration is positive.
\end{definition}

\begin{definition}
A CRN is said to \emph{admit multiple (positive) equilibria} or is \emph{multi-stationary} if there exist positive rate constants such that the ODE system admits more than one stoichiometrically-compatible equilibria.
\end{definition}


\subsection{Embedded networks and network transformations}
\label{sec:embed}

In this section, we briefly review the concepts of embedded networks and network transformations.
First, the definition of embedded networks from \cite{JOSH2013} is as follows:

\begin{definition}
    Let $\mathscr{N} = \{ \mathscr{S}, \mathscr{C}, \mathscr{R} \}$ be a CRN.
    \begin{enumerate}
        \item Consider a subset of the species set $S \subset \mathscr{S}$, a subset of the complexes set $C \subset \mathscr{C}$, and a subset of the reactions set $R \subset \mathscr{R}$.
        \begin{itemize}
            \item The \emph{restriction of $R$ to $S$}, denoted $R\vert_S$, is the set of reactions obtained by taking the reactions in $R$ and removing all species not in $S$ from the reactant and product complexes. If a trivial reaction (one in which the reactant and product complexes are the same) is obtained in this process, then it is removed. Extra copies of repeated reactions were also removed. The \emph{restriction of $C$ to $R$}, denoted $C\vert_R$, is the set of (reactant and product) complexes of reactions in $R$.
            \item The \emph{restriction of $S$ to $C$}, denoted $S\vert_C$, is the set of species that are in the complexes in $C$.
        \end{itemize}
        \item The network obtained from $\mathscr{N}$ by \emph{removing a subset of species} $\{ X_i \} \subset S$ is the network
        $ \left\{ \mathscr{S} \backslash \{ X_i \}, \mathscr{C}\vert_{\mathscr{R}\vert_{\mathscr{S} \backslash \{ X_i \}}}, \mathscr{R}\vert_{\mathscr{S} \backslash \{ X_i \} } \right\}. $
        \item A subset of the reactions $\mathscr{R}' \subset \mathscr{R}$ defines the \emph{subnetwork} $$\{ \mathscr{S}\vert_{\mathscr{C}\vert_{\mathscr{R}'}}, \mathscr{C}\vert_{\mathscr{R}'}, \mathscr{R}' \}.$$
        \item Let $\mathscr{N} = \{ \mathscr{S}, \mathscr{C}, \mathscr{R} \}$ be a CRN. An \emph{embedded network} of $\mathscr{N}$ is defined by a subset of the species set {$S = \{ X_{i_1}, X_{i_2}, \dots, X_{i_k} \} \subset \mathscr{S}$}, and a subset of the reactions set $R = \{ R_{j_1}, R_{j_2}, \dots, R_{j_l} \} \subset \mathscr{R}$, that involve all species of $S$, is the network $(S, \mathscr{C}\vert_{R\vert_S}, R\vert_S)$ consisting of the reactions $R\vert_S$.
    \end{enumerate}
\end{definition}

\begin{example}
    
    Models in BST have two types of variables: dependent (state variables that change with time) and independent (constants that describe aspects of the process environment). Arceo et al. \cite{AJLM2015} introduced a power law kinetic realization, whereby the species set of the underlying network---called a total representation---consists of the corresponding dependent and independent species. In a subsequent work \cite{AJLM2017}, they also defined the embedded representation of a BST model, which is the embedded network of the total representation relative to the dependent species. The embedded representation corresponds to the popular technique in BST of ``lumping the independent variables to the rate constants''.
\end{example}

The concept of network transformation originated in Nazareno et al. \cite{NEML2019} and was recently extended by Talabis and Mendoza \cite{TAME2023} in their study of network operations. The most well-known example of network transformation is the network translation of mass action systems introduced by Johnston in 2014 \cite{JOHN2014}. Hong et al. recently showed that any network translation can be composed from three network operations: shifting, dividing, and merging \cite{HHLL2022}. In this work, we will use shifting and ``splitting by reaction vector'' to define the network transformations necessary for our study.

We first recall the definition of dynamical equivalence of kinetic systems:

\begin{definition}
     Two kinetic systems $(\mathscr{N}, K)$ and $(\mathscr{N}', K')$ with stoichiometric matrix $N$ and $N'$, respectively, are \emph{dynamically equivalent} if
     \begin{enumerate}
         \item[i.] {$\mathscr{N}$ and $\mathscr{N}'$ have the same set of species;}
         \item[ii.] $K$ and $K'$ have the same definition domain $\Omega = \Omega'$; and 
         \item[iii.] $N K(x) = N' K'(x)$ for all $x \in \Omega$.
     \end{enumerate}
\end{definition}

The last condition implies that the ODE systems of $\mathscr{N}$ and $\mathscr{N}'$ are identical.

 We now define several classes of network transformations:

 \begin{definition}
     A dynamical equivalence between $(\mathscr{N}, K)$ {(with stoichiometric subspace $S$)} and $(\mathscr{N}', K')$ {(with stoichiometric subspace $S'$)} is called
     \begin{enumerate}
         \item[i.] \emph{$S$-extending transformation} if $S$ is contained in $S'$.
         \item[ii.] \emph{$S$-including transformation} if $S$ contains $S'$.
         \item[iii.] \emph{$S$-invariant transformation} if $S = S'$.
         \item[iv.] \emph{Network transformation} (or simply \emph{transformation}) if one of $(i)$, $(ii)$, or $(iii)$ holds.
     \end{enumerate}
 \end{definition}

 Network operations are network transformations that change a single reaction or a pair of reactions in a network. Hence, they can be viewed as ``building blocks'' of transformations. Two useful operations are as follows:

\begin{definition}
\label{def:shifting:splitting}
    \begin{enumerate}
        \item \emph{Shifting} the reaction {$q:y\to y'$} results to $q': y + z \xrightarrow{r} y' + z$ with kinetics $K_{q'} = K_q$ where $z = z_1 X_1 + \dots + z_m X_m$ and $z_i \in \mathbb{Z}$ for all $i = 1, \dots, m$.
        \item \emph{Splitting} [via reaction vector] (RV-splitting) the reaction {$q:y\to y'$} results to the pair of reactions $q': x \xrightarrow{r} x'$ and $q'': z \xrightarrow{r} z'$ with the same kinetics (i.e., $K_q = K_{q'} = K_{q''}$) and $y' - y = (x' - x) + (z' - z)$.
    \end{enumerate}
\end{definition}

Note that shifting is an $S$-invariant transformation while splitting is an $S$-extending one. More examples and applications of network operations and transformations can be found in \cite{TAME2023}.

\section{The method of comparing reaction networks via common species embedded networks}
\label{CSEN:method}

In this section, we introduce a new method for comparing reaction networks of the same or closely related biological systems. We call this method the {\emph{Common Species Embedded Networks (CSEN)}} analysis.
This approach allows us to determine a possible relationship between embedded networks of two CRNs via network operations or transformations, such as dynamical equivalence.

{Before we formulate the steps of the CSEN analysis, we introduce some new concepts as follows.

\begin{definition}
    \label{def:proximate}
    Let $(\mathscr{N}, K)$ and $(\mathscr{N}^*, K^*)$ be kinetic systems. They are \emph{proximately equivalent} if there are kinetics $\widetilde{K}$ and $\widetilde{K}^*$, differing (at most in inflows) from $K$ and $K^*$, respectively, such that $(\mathscr{N}, \widetilde{K})$ is dynamically equivalent to $(\mathscr{N}^*, \widetilde{K}^*)$. When the dynamical equivalence between $(\mathscr{N}, \widetilde{K})$ and $(\mathscr{N}^*, \widetilde{K}^*)$ is a network transformation, we speak of a \emph{proximate transformation} from $(\mathscr{N}, K)$ to $(\mathscr{N}^*, K^*)$.
\end{definition}

\begin{remark}
 \begin{enumerate}
         \item Two kinetic systems differ at most in inflows if they share the same set of reactions, but one system includes at least one inflow $0 \to A_i$ that is absent in the other.
        \item The domain of the chosen proximate transformation is the set of proximate reactions of $\mathscr{N}_E$. Its image is the set of proximate reactions of $\mathscr{N}_E^*$.
\end{enumerate}
\end{remark}

\begin{example}
    \begin{enumerate}
        \item[i.] Dynamically equivalent systems are proximately equivalent.
        \item[ii.] Let $\mathscr{N}_E$ and $\mathscr{N}_E^*$ be CSEN. A maximal proximate transformation from $\mathscr{N}_E$ to $\mathscr{N}_E^*$ is a proximate transformation defined by a subset of non-common reactions of $\mathscr{N}_E$ with the maximal number of elements. Note that there may be more than one such subset; hence, more than one maximal proximate transformation.
    \end{enumerate}
\end{example}

}

We now carry out the method in three steps.
The first step is to construct the embedded networks $\mathscr{N}_E$ and $\mathscr{N}_E^*$ of two networks $\mathscr{N}$ and $\mathscr{N}^*$, respectively, with respect to their set of common species $\mathscr{S}' = \mathscr{S}_\mathscr{N} \cap \mathscr{S}_{\mathscr{N}^*}$.

In the second step, we classify the reactions of each embedded network into 3 subsets or reaction classes:
\begin{enumerate}
    \item Common reactions of $\mathscr{N}$ and $\mathscr{N}^*$: $\mathscr{R}_\mathscr{N} \cap \mathscr{R}_{\mathscr{N}^*}$
    \item Common reactions of the embedded networks which are not in the first class: $(\mathscr{R}_{\mathscr{N}_E} \cap \mathscr{R}_{\mathscr{N}_E^*}) \backslash (\mathscr{R}_\mathscr{N} \cap \mathscr{R}_{\mathscr{N}^*})$
    \item Reactions of $\mathscr{N}_E$ and $\mathscr{N}_E^*$ for further classification into the proximate and non-equivalent subsystems (see Definition \ref{def:proximate})
\end{enumerate}


{The third step consists of the following task:
identify a maximal proximate transformation from $\mathscr{N}_E$ to $\mathscr{N}_E^*$ (if it exists).}

\section{CSEN analysis of the Lee and Feinberg models}
\label{sec:comparisonLeeFAL}

In this section, we compare the reaction networks of the Lee and Feinberg models by first constructing their embedded networks with respect to the set of common species. We then identify network operations that transform the remaining distinct reactions of the Lee network into those of the Feinberg network. Our result shows that the two networks are dynamically very similar and points to a relative aspect of the ``model discrimination'' by MacLean et al. \cite{Maclean} into mono- and multi-stationary systems.



We now illustrate the approach using the Lee and Feinberg networks.

Since the species sets of the Lee and Feinberg networks are $$\mathscr{S}_L = \{A_1,A_2,A_4,A_6,A_7,A_8,A_{10},A_{12},A_{13},A_{22},\ldots,A_{28}\} \backslash \{A_{28}\}$$ and $$\mathscr{S}_F = \{A_1,A_2,A_4,A_6,A_7,A_8,A_{10},A_{12},A_{13},A_{22},\ldots,A_{28}\} \backslash \{A_{22}\},$$ then both networks have 15 species with 14 in common. Table \ref{tab:reactionsNLEandNFE} presents the results of Steps 1 and 2 of the CSEN analysis.

\begin{table}[ht!]
    \centering
    \caption{Reactions of the embedded networks ($\mathscr{N}_{LE}$ and $\mathscr{N}_{FE}$) of the Lee ($\mathscr{N}_{L}$) and Feinberg ($\mathscr{N}_{F}$) models with respect to their common species. Reaction numbers without the superscript $E$ do not change after getting the embedded networks with respect to their common species.}
    \label{tab:reactionsNLEandNFE}
    \begin{tabular}{|l|l|}
        \hline
        \multicolumn{2}{|c|}{Common to $\mathscr{N}_{LE}$ and $\mathscr{N}_{FE}$} \\
        \hline
        \multicolumn{2}{|l|}{\hspace{2.5 cm} $R_1: 0 \rightarrow A_4$} \\
        \multicolumn{2}{|l|}{\hspace{2.5 cm} $R_4: A_1 + A_4 \rightarrow A_8$} \\
        \multicolumn{2}{|l|}{\hspace{2.5 cm} $R_5: A_8 \rightarrow A_1  +  A_4$} \\
        \multicolumn{2}{|l|}{\hspace{2.5 cm} $R_{12}: A_{10} \rightarrow 0$} \\
        \multicolumn{2}{|l|}{\hspace{2.5 cm} $R_{14}: A_1 \rightarrow A_2$} \\
        \multicolumn{2}{|l|}{\hspace{2.5 cm} $R_{15}: A_2 \rightarrow A_1$} \\
        \multicolumn{2}{|l|}{\hspace{2.5 cm} $R_{18}: A_{12} \rightarrow A_{13}$} \\
        \multicolumn{2}{|l|}{\hspace{2.5 cm} $R_{19}: A_{13} \rightarrow A_{12}$} \\
        \multicolumn{2}{|l|}{\hspace{2.5 cm} $R_{38}: A_4 \rightarrow 0$} \\
        \multicolumn{2}{|l|}{\hspace{2.5 cm} $R_{43}: A_{24} + A_{26} \rightarrow A_{23}$} \\
        \multicolumn{2}{|l|}{\hspace{2.5 cm} $R_{44}: A_{23} \rightarrow A_{24} + A_{26}$} \\
        \multicolumn{2}{|l|}{\hspace{2.5 cm} $R_{45}: A_8 \rightarrow A_{25}$} \\
        \multicolumn{2}{|l|}{\hspace{2.5 cm} $R_{46}: A_{25} \rightarrow A_1 + A_{10}$} \\
        \multicolumn{2}{|l|}{\hspace{2.5 cm} $R_{47}: 0 \rightarrow A_{26}$} \\
        \multicolumn{2}{|l|}{\hspace{2.5 cm} $R_{48}: A_{26} \rightarrow 0$} \\
        \multicolumn{2}{|l|}{\hspace{2.5 cm} $R_{49}: A_4 + A_6 \rightarrow A_7$} \\
        \multicolumn{2}{|l|}{\hspace{2.5 cm} $R_{50}: A_7 \rightarrow A_4 + A_6$} \\
        \multicolumn{2}{|l|}{\hspace{2.5 cm} $R_{51}: A_{24} + A_4 \rightarrow A_{27}$} \\
        \multicolumn{2}{|l|}{\hspace{2.5 cm} $R_{52}: A_{27} \rightarrow A_{24} + A_4$} \\
        \hline
        \hline
        \multicolumn{2}{|c|}{Embedding-derived common reactions} \\
        \hline
        \multicolumn{2}{|l|}{\hspace{2.5 cm} $R_{41}^E:  A_{23} \rightarrow A_2$} \\ 
        \multicolumn{2}{|l|}{\hspace{2.5 cm} $R_{42}^E: A_2 \rightarrow A_{23}$} \\ 
        \hline
        \hline
        \multicolumn{1}{|c|}{Unique to $\mathscr{N}_{LE}$} & \multicolumn{1}{c|}{Unique to $\mathscr{N}_{FE}$} \\
        \hline
        $R_{40}^E: A_{13} + A_2 \rightarrow A_{13} + A_{23}$ & $R_{53}^E: A_{13} + A_2 \rightarrow 0$ \\
         & $R_{56}^E: 0 \rightarrow A_{13} + A_{23}$ \\
        \hline
    \end{tabular}
\end{table}

\begin{theorem}
    \label{thm:4.8}
    Let $\mathscr{N}_{LE}$ and $\mathscr{N}_{FE}$ be the embedded networks of $\mathscr{N}_{L}$ and $\mathscr{N}_{F}$, respectively, with respect to their set of common species. Then $\mathscr{N}_{LE}$ with mass action kinetics is dynamically equivalent to $\mathscr{N}_{FE}$ with mass action kinetics except for an inflow reaction. The dynamical equivalence is via an S-extending network operation.
\end{theorem}

\begin{proof}
    The common reactions of $\mathscr{N}_{LE}$ and $\mathscr{N}_{FE}$ consist of the 21 reactions listed in Table \ref{tab:reactionsNLEandNFE}. The other remaining reactions in Feinberg become $A_{13} + A_2 \to 0 \to A_{13} + A_{23}$. This is the transform by RV-splitting (an $S$-extending network operation) from $A_{13} + A_2 \to A_{13} + A_{23}$ in the Lee network. Furthermore, the rank of $\mathscr{N}_{LE}$ is one less than that of $\mathscr{N}_{FE}$. 
\end{proof}

\begin{remark}
    \begin{enumerate}
        \item {The aberrant kinetics on the inflow $0 \to A_{13} + A_{23}$ is given by $kA_{13}A_{2}$ (versus the usual constant kinetics convention). The choice of this kinetics is due to splitting the reaction $A_{13}+A_2 \to A_{13}+A_{23}$ (with the usual kinetics $kA_{13}A_{2}$) into $A_{13}+A_2 \to 0$ and $0 \to A_{13}+A_{23}$ (with the same kinetics) as seen, in general, in Definition \ref{def:shifting:splitting}.} 
        \item The expansion of the stoichiometric subspace potentially explains the mono-/multi-stationarity difference, as equilibria previously in separate stoichiometric classes can fall into the same larger class. Since the embedded networks describe the projection of processes in the larger (different) 15-dimensional species space onto the common 14-dimensional subspace, our result shows the strong dynamical similarity between the Lee and Feinberg systems.
    \end{enumerate}
\end{remark}

For Step 3 of the CSEN analysis, it is shown in the proof of Theorem \ref{thm:4.8} that the subnetworks of the non-common reactions are proximately equivalent. Thus, the embedded networks of the Lee and Feinberg models are proximately equivalent. In other words, there is a proximate transformation from $\mathscr{N}_{LE}$ to $\mathscr{N}_{FE}$.

\section{CSEN analysis of the Schmitz and MacLean models}
\label{section:Schmitz:MacLean}


\label{section:CSEN}

The coincidence of many structo-kinetic (e.g., multi-stationarity) and kinetic (e.g., ACR) properties of the Schmitz network ($\mathscr{N}_{S}$) and the MacLean network ($\mathscr{N}_{M}$) suggests a structural relationship between them. In this section, we use CSEN analysis to derive such a relationship between their embedded networks, denoted by $\mathscr{N}_{SE}$ and $\mathscr{N}_{ME}$, induced by the subset of common species. More precisely, we show that
\begin{itemize}
    \item A network transformation of $\mathscr{N}_{SE}$ coincides with $\mathscr{N}_{MEA}$ ($:= \mathscr{N}_{ME}$ less 2 outflows plus 1 reversible pair of ``flow'' reactions); and
    \item The kinetics of $\mathscr{N}_{SE}$ is mass action, and that of $\mathscr{N}_{ME}$, while generalized mass action, is mass action except for 3 ``flow'' reactions.
\end{itemize}

\begin{remark}
In mass action systems, the rate function of a reaction is based on the stoichiometric coefficients in the reactant complex. For example, under mass action kinetics, the reaction $A_5 + A_3 \to A_9$ has the rate function $k a_5 a_3$.
However in generalized mass action systems, introduced by Stefan M\"{u}ller and Georg Regensburger \cite{SMGR2012}, the rate function is not necessarily based on the stoichiometric coefficients in the reactant complex. In the example above, under generalized mass action kinetics (GMAK), the rate function can have the form $k a_5^\alpha a_3^\beta$ where $\alpha$ and $\beta$ can be any positive real number.
\end{remark}

Table \ref{tab:reactionsNSEandNME} shows the results of the CSEN analysis Steps 1 and 2.

\begin{table}[ht!]
    \centering
    \caption{Reactions of the embedded networks ($\mathscr{N}_{SE}$ and $\mathscr{N}_{ME}$) of the Schmitz ($\mathscr{N}_{S}$) and MacLean ($\mathscr{N}_{M}$) models with respect to their common species. Reactions without the superscript $E$ do not change after getting the embedded networks with respect to the common species.}
    \label{tab:reactionsNSEandNME}
    \begin{tabular}{|l|l|}
        \hline
        \multicolumn{2}{|c|}{Common to $\mathscr{N}_{SE}$ and $\mathscr{N}_{ME}$} \\
        \hline
        \multicolumn{2}{|l|}{\hspace{1.5 cm} $R_1: 0 \rightarrow A_4$} \\
        \multicolumn{2}{|l|}{\hspace{1.5 cm} $R_2: A_4 \rightarrow A_5$} \\
        \multicolumn{2}{|l|}{\hspace{1.5 cm} $R_3: A_5 \rightarrow A_4$} \\
        \multicolumn{2}{|l|}{\hspace{1.5 cm} $R_4: A_1 + A_4 \rightarrow A_8$} \\
        \multicolumn{2}{|l|}{\hspace{1.5 cm} $R_5: A_8 \rightarrow A_1 + A_4$} \\
        \multicolumn{2}{|l|}{\hspace{1.5 cm} $R_6: A_5 + A_3 \rightarrow A_9$} \\
        \multicolumn{2}{|l|}{\hspace{1.5 cm} $R_7: A_9 \rightarrow A_5 + A_3$} \\
        \multicolumn{2}{|l|}{\hspace{1.5 cm} $R_8: A_6 + A_5 \rightarrow A_7$} \\
        \multicolumn{2}{|l|}{\hspace{1.5 cm} $R_9: A_7 \rightarrow A_6 + A_5$} \\
        \hline
        \hline
        \multicolumn{2}{|c|}{Embedding-derived common reactions} \\
        \hline
        \multicolumn{2}{|l|}{\hspace{1.5 cm} $R_{10}^E: A_8 \rightarrow A_1$} \\ 
        \multicolumn{2}{|l|}{\hspace{1.5 cm} $R_{11}^E: A_9 \rightarrow A_3$} \\ 
        \hline
        \hline
        \multicolumn{1}{|c|}{Unique to $\mathscr{N}_{SE}$} & \multicolumn{1}{c|}{Unique to $\mathscr{N}_{ME}$} \\
        \hline
        $R_{14}: A_1 \rightarrow A_2$ & $R_{38}: A_4 \rightarrow 0$ \\
        $R_{15}: A_2 \rightarrow A_1$ & $R_{39}: A_5 \rightarrow 0$ \\
        $R_{16}: A_1 \rightarrow A_3$ & $R_{22}^E: A_2 \rightarrow 0$ \\ 
        $R_{17}: A_3 \rightarrow A_1$ & $R_{23}^E: 0 \rightarrow A_2$ \\ 
            & $R_{24}^E: A_3 \rightarrow 0$ \\
            & $R_{25}^E: 0 \rightarrow A_3$ \\ 
            & $R_{30}^E: A_1 \rightarrow 0$ \\
            & $R_{31}^E: 0 \rightarrow A_1$ \\ 
        \hline
    \end{tabular}
\end{table}

For the embedded Schmitz network, the common species comprise a large subset (9 in $\mathscr{N}_{ME}$ out of the 11 in $\mathscr{N}_S$) and the reactions generate a large subnetwork ($\mathscr{N}_{ME}$ has rank 7 while $\mathscr{N}_S$ has rank 9). This  observation already suggests an ``inclusion/extension-like'' relationship between the networks.

The set of new common reactions is $(\mathscr{R}_{SE} \cap  \mathscr{R}_{ME}) \backslash (\mathscr{R}_S \cap \mathscr{R}_M) = \{ A_8 \rightarrow A_1, A_9 \rightarrow A_3 \}$. Hence, the total of common reactions make up 11 of the 15 reactions in $\mathscr{N}_{SE}$.

The results of Step 3 are as follows: the following table defines two maximal proximate transformation from $\mathscr{N}_{SE}$ to $\mathscr{N}_{ME}$:
\begin{center}
    \begin{tabular}{|l|l|l|}
        \hline
        \multicolumn{1}{|c|}{Reaction} & \multicolumn{1}{c|}{Operation} & \multicolumn{1}{c|}{Parameter} \\
        \hline
        $A_1 \rightarrow A_2$ & Splitting & $A_2 - A_1 = (A_2 - 0) + (0 - A_1)$ \\
        \hline
        $A_2 \rightarrow A_1$ & Splitting & $A_1 - A_2 = (A_1 - 0) + (0 - A_2)$ \\
        \hline
        $A_1 \rightarrow A_3$ & Splitting & $A_3 - A_1 = (A_3 - 0) + (0 - A_1)$ \\
        \hline
        $A_3 \rightarrow A_1$ & Splitting & $A_1 - A_3 = (A_1 - 0) + (0 - A_3)$ \\
        \hline
    \end{tabular}
\end{center}

The top two rows define the first maximal proximate transformation, the bottom two the second. We choose the first and, hence, consider the reactions $\{ A_3 \rightleftarrows A_1 \}$ and $\{ A_3 \rightleftarrows 0, A_4 \rightarrow 0, A_5 \rightarrow 0 \}$ as the non-equivalent reactions of the CSEN.

Defining the augmented MacLean CSEN as $\mathscr{N}_{MEA} := \mathscr{N}_{ME} \cup \{ A_1 \rightleftarrows 2 A_1 \}$, we can extend the previous considerations to a ``structural-kinetic relationship'' as follows:

\begin{theorem}
    The (mass action) system $\mathscr{N}_{SE}$ transforms to the subnetwork $\mathscr{N}_{MEA} \backslash \{ A_4 \rightarrow 0, A_5 \rightarrow 0 \}$, a generalized mass action system, which is predominantly mass action (14 of 19 reactions). Moreover, the transformation is the identity in 11 of the 15 reactions of $\mathscr{N}_{SE}$.
\end{theorem}

\begin{proof}
    The following table describes the GMAK details. We set $K_i = k_i A_i$ where $k_i > 0$, i.e., linear kinetics for the species $A_i$:
    \begin{center}
    \begin{tabular}{|l|l|l|}
        \hline
        \multicolumn{1}{|c|}{Reaction} & \multicolumn{1}{c|}{Transformation 1/} & \multicolumn{1}{c|}{Transformation 2/} \\
        & \multicolumn{1}{c|}{ Kinetics/Type} & \multicolumn{1}{c|}{ Kinetics/Type} \\
        \hline
        $A_1 \rightarrow A_2$ & $0 \rightarrow A_2$/ $K_1$/ GMAK & $A_1 \rightarrow 0$/ $K_1$/ MAK \\
        \hline
        $A_2 \rightarrow A_1$ & $0 \rightarrow A_1$/ $K_2$/ GMAK & $A_2 \rightarrow 0$/ $K_2$/ MAK \\
        \hline
        $A_1 \rightarrow A_3$ & $0 \rightarrow A_3$/ $K_1$/ GMAK & $A_1 \rightarrow 0$/ $K_1$/ MAK \\
        \hline
        $A_3 \rightarrow A_1$ & $0 \rightarrow A_1$/ $K_3$/ GMAK & $A_3 \rightarrow 0$/ $K_3$/ MAK \\
        \hline
        $0 \rightarrow A_1$ & $A_1 \rightarrow 2 A_1$/ $K_3$/ GMAK & \\
        \hline
        $A_1 \rightarrow 0$ & $2 A_1 \rightarrow A_1$/ $K_3$/ GMAK & \\
        \hline
    \end{tabular}
    \end{center}
    The rest of the reactions (and their mass action kinetics) in $\mathscr{N}_{SE}$ remain unchanged. This proves the claim that $\mathscr{N}_{MEA}$ is a network transformation of $\mathscr{N}_{SE}$.
\end{proof}

\section{CSEN analysis of the Feinberg and\\MacLean models}
\label{section:comp:FAL:MacLean}

Table \ref{tab:reactionsNFandNM} presents the species and reactions of the Feinberg ($\mathscr{N}_F$) and MacLean ($\mathscr{N}_M$) models. The Feinberg and MacLean networks have eight common species: $A_1$, $A_2$, $A_4$, $A_6$, $A_7$, $A_8$, $A_{12}$, and $A_{13}$, and six common reactions: $R_1$, $R_4$, $R_5$, $R_{18}$, $R_{19}$ and $R_{38}$. Following the discussion in Section \ref{sec:embed}, the embedded networks of the two models induced by their common species are given in Table \ref{tab:reactionsNFEandNME}.

\begin{table}[ht!]
	\centering
	\caption{Species and reactions of the Feinberg ($\mathscr{N}_F$) and MacLean ($\mathscr{N}_M$) Wnt signaling models}
	\label{tab:reactionsNFandNM}
	\begin{tabular}{|l|l|}
		\hline
		\multicolumn{2}{|c|}{Common to $\mathscr{N}_F$ and $\mathscr{N}_M$} \\
            \hline
            \multicolumn{2}{|c|}{$A_1$, $A_2$, $A_4$ $A_6$, $A_7$, $A_8$, $A_{12}$, $A_{13}$} \\
		\hline
		\multicolumn{2}{|l|}{\hspace{2 cm} $R_1: 0 \rightarrow A_4$} \\
		\multicolumn{2}{|l|}{\hspace{2 cm} $R_4: A_1+A_4 \rightarrow A_8$} \\
		\multicolumn{2}{|l|}{\hspace{2 cm} $R_5: A_8 \rightarrow A_1 + A_4$} \\
		\multicolumn{2}{|l|}{\hspace{2 cm} $R_{18}: A_{12} \rightarrow A_{13}$} \\
		\multicolumn{2}{|l|}{\hspace{2 cm} $R_{19}: A_{13} \rightarrow A_{12}$} \\
		\multicolumn{2}{|l|}{\hspace{2 cm} $R_{38}: A_4 \rightarrow 0$} \\
		\hline
            \hline
		\multicolumn{1}{|c|}{Unique to $\mathscr{N}_F$} & \multicolumn{1}{c|}{Unique to $\mathscr{N}_M$} \\
            \hline
            \multicolumn{1}{|c|}{$A_{10}$, $A_{23}$, \dots, $A_{28}$} & \multicolumn{1}{c|}{$A_3$, $A_5$, $A_9$, $A_{14}$, \dots, $A_{21}$} \\
		\hline
        $R_{12}: A_{10} \rightarrow 0$ & $R_{2}: A_{4} \rightarrow A_{5}$ \\
		$R_{14}: A_1 \rightarrow A_2$ & $R_{3}: A_{5} \rightarrow A_{4}$ \\
		$R_{15}: A_{2} \rightarrow A_1$ & $R_{6}: A_{5} +A_3 \rightarrow A_{9}$ \\
		$R_{43}: A_{24} + A_{26} \rightarrow A_{23}$ & $R_{7}: A_{9} \rightarrow A_{5}+A_3$ \\
		$R_{44}: A_{23} \rightarrow A_{24} + A_{26}$ & $R_{8}: A_6 +A_5 \rightarrow A_{7}$ \\
		$R_{45}: A_8 \rightarrow A_{25}$ & $R_{9}: A_{7} \rightarrow A_6 + A_5$ \\
		$R_{46}: A_{25} \rightarrow A_1 + A_{10}$ & $R_{20}: A_{13} \rightarrow A_{14}$ \\
		$R_{47}: 0 \rightarrow A_{26}$ & $R_{21}: A_{14} \rightarrow A_{13}$ \\
		$R_{48}: A_{26} \rightarrow 0$ & $R_{22}: A_2 \rightarrow A_{15}$ \\
		$R_{49}: A_4 + A_6 \rightarrow A_7$ & $R_{23}: A_{15} \rightarrow A_2$ \\
		$R_{50}: A_7 \rightarrow A_4 + A_6$ & $R_{24}: A_3 + A_{14} \rightarrow A_{19}$ \\
		$R_{51}: A_{24} + A_4 \rightarrow A_{27}$ & $R_{25}: A_{19} \rightarrow A_3 + A_{14}$ \\
        $R_{52}: A_{27} \rightarrow A_{24} + A_4$ & $R_{26}: A_{19} \rightarrow A_{14} + A_{15}$ \\
        $R_{53}: A_{13} + A_2 \rightarrow A_{28}$
    	& $R_{27}: A_{15} + A_{17} \rightarrow A_{21}$ \\
     $R_{54}: A_2 \rightarrow A_{23}$
    	& $R_{28}: A_{21} \rightarrow A_{15} + A_{17}$ \\
     $R_{55}: A_{23} \rightarrow A_2$
    	& $R_{29}: A_{21} \rightarrow A_3 + A_{17}$ \\
     $R_{56}: A_{28} \rightarrow A_{13} + A_{23}$
            & $R_{30}: A_{13} + A_1 \rightarrow A_{18}$ \\
    	& $R_{31}: A_{18} \rightarrow A_{13} +A_1$ \\
    	& $R_{32}: A_{18} \rightarrow A_{13} +A_2$ \\
    	& $R_{33}: A_2 + A_{16} \rightarrow A_{20}$ \\
    	& $R_{34}: A_{20} \rightarrow A_2 + A_{16}$ \\
    	& $R_{35}: A_{20} \rightarrow A_1 + A_{16}$ \\
    	& $R_{36}: A_8 \rightarrow A_1$ \\
    	& $R_{37}: A_9 \rightarrow A_3$ \\
    	& $R_{39}: A_5 \rightarrow 0$ \\
		\hline
	\end{tabular}
\end{table}

\begin{table}[ht!]
    \centering
    \caption{Reactions of the embedded networks ($\mathscr{N}_{FE}$ and $\mathscr{N}_{ME}$) of the Feinberg ($\mathscr{N}_{F}$) and MacLean ($\mathscr{N}_{M}$) models with respect to their common species. Reactions without the superscript $E$ do not change after getting the embedded networks with respect to the common species.}
	\label{tab:reactionsNFEandNME}
	\begin{tabular}{|l|l|}
		\hline
		\multicolumn{2}{|c|}{Common to $\mathscr{N}_{FE}$ and $\mathscr{N}_{ME}$} \\
		\hline
		\multicolumn{2}{|l|}{\hspace{2 cm} $R_1: 0 \rightarrow A_4$} \\ 
		\multicolumn{2}{|l|}{\hspace{2 cm} $R_4: A_1+A_4 \rightarrow A_8$} \\
		\multicolumn{2}{|l|}{\hspace{2 cm} $R_5: A_8 \rightarrow A_1 + A_4$} \\
		\multicolumn{2}{|l|}{\hspace{2 cm} $R_{18}: A_{12} \rightarrow A_{13}$} \\
		\multicolumn{2}{|l|}{\hspace{2 cm} $R_{19}: A_{13} \rightarrow A_{12}$} \\
		\multicolumn{2}{|l|}{\hspace{2 cm} $R_{38}: A_4 \rightarrow 0$} \\ 
            \hline
            \hline
            \multicolumn{2}{|c|}{Embedding-derived common reactions} \\
            \hline
    	\multicolumn{2}{|l|}{\hspace{2 cm} $R_{46}^E:  0 \rightarrow A_1$} \\ 
    	\multicolumn{2}{|l|}{\hspace{2 cm} $R_{54}^E: A_2 \rightarrow {0}$} \\ 
        \multicolumn{2}{|l|}{\hspace{2 cm} $R_{55}^E:  0 \rightarrow A_2$} \\ 
        \multicolumn{2}{|l|}{\hspace{2 cm} $R_{56}^E:  0 \rightarrow A_{13}$} \\ 
		\hline
            \hline
		\multicolumn{1}{|c|}{Unique to $\mathscr{N}_{FE}$} & \multicolumn{1}{c|}{Unique to $\mathscr{N}_{ME}$} \\
		\hline
		$R_{14}: A_{1} \rightarrow A_{2}$ & $R_{8}^E: A_{6} \rightarrow A_7$ \\
        $R_{15}: A_2 \rightarrow A_{1}$ & $R_{9}^E: A_{7} \rightarrow A_6$\\
        $R_{45}^E: A_8 \rightarrow 0$ & $R_{20}^E: A_{13} \rightarrow 0$\\
        $R_{49}: A_4 + A_6 \rightarrow A_{7}$ & $R_{30}^E: A_{13}+A_1 \rightarrow 0$\\
        $R_{50}: A_7 \rightarrow A_4 + A_6$ & $R_{31}^E: 0 \rightarrow A_{13}+A_1$\\
        $R_{53}^E: A_{13} + A_2 \rightarrow 0$ & $R_{32}^E: 0 \rightarrow A_{13}+A_2$\\
        & $R_{36}^E: A_8 \rightarrow A_{1}$\\
		\hline
	\end{tabular}
\end{table}




Despite extensive systematic search, we were unable to identify any proximate equivalences. This points to a limitation of the relationship of the two embedded networks to the purely structural coincidence of 10 of the 15--16 reactions.

\section{Summary and recommendation}
\label{sec:summary}
We presented a novel method, which we call the Common Species Embedded Networks (CSEN) analysis, to compare reaction networks in closely related systems. We applied this approach to assess existing reaction networks associated with Wnt signaling models (Lee, Schmitz, MacLean, and Feinberg).

The CSEN analysis produced three interesting results.
Firstly, the analysis revealed a strong similarity between the Lee (mono-stationary) and Feinberg (multi-stationary) models. This result challenges the absolute discrimination of models into mono-stationarity and multi-stationarity by Maclean et al.
Second, our analysis provided evidence supporting similarity between the Schmitz and MacLean models, highlighted by the identified ``proximate equivalence."
Lastly, the analysis emphasized the absence of a comparable relationship between the Feinberg and MacLean models, emphasizing distinctive differences between the two. {Situations where the search for proximate transformations become challenging (as was the case in Section \ref{section:comp:FAL:MacLean}) present further research opportunities.}

The CSEN analysis introduced in this study offers a valuable tool for comparing mathematical models in related systems. This approach may be useful in uncovering structural similarities and differences between reaction networks. Further exploration and application of the CSEN analysis could enhance our understanding of the dynamics of closely related systems and contribute to refining mathematical models in biological and chemical contexts.


\section*{Acknowledgement}
BSH acknowledge the Personally Funded Research support from the Institute of Mathematics, College of Science, and the Office of the Vice Chancellor for Research and Development of the University of the Philippines Diliman. BSH was awarded the MacArthur \& Josefina De los Reyes Professorial Chair in Mathematics of the 2023 UPFI Inc. Professorial Chairs and Faculty Grants (OVCAA, UP Diliman) for this project.



\appendix
\section{Definition of variables}
\label{app:var}

Table \ref{tab:species} provides the species in the Wnt signaling networks that we consider.

\begin{table}[ht!]
\centering
\caption{Species and corresponding biomolecules that can occur in the Wnt signaling models considered}
\label{tab:species}
\begin{tabular}{|c|l|}
    \hline
    Species & \multicolumn{1}{c|}{Meaning} \\
    \hline
    $A_1$ & destruction complex (DC) (active form) \\
    \hline
    $A_2$ & DC (inactive form) \\
    \hline
    $A_3$ & active DC residing in the nucleus \\
    \hline
    $A_4$ & $\beta$-catenin \\
    \hline
    $A_5$ & $\beta$-catenin in the nucleus \\
    \hline
    $A_6$ & T-cell factor (TCF) \\
    \hline
    $A_7$ & $\beta$-catenin-TCF complex \\
    \hline
    $A_8$ & $\beta$-catenin bound with DC \\
    \hline
    $A_9$ & $\beta$-catenin bound with DC in the nucleus \\
    \hline
    $A_{10}$ & $\beta$-catenin (for proteasomal degradation) \\
    \hline
    $A_{11}$ & $\beta$-catenin (for proteasomal degradation) in the nucleus \\
    \hline
    $A_{12}$ & dishevelled (inactive form) \\
    \hline
    $A_{13}$ & dishevelled (active form) \\
    \hline
    $A_{14}$ & active dishevelled in the nucleus \\
    \hline
    $A_{15}$ & inactive DC in the nucleus \\
    \hline
    $A_{16}$ & phosphatase \\
    \hline
    $A_{17}$ & phosphatase in the nucleus \\
    \hline
    $A_{18}$ & active DC bound with dishevelled \\
    \hline
    $A_{19}$ & active DC bound with dishevelled in the nucleus \\
    \hline
    $A_{20}$ &active DC bound with phosphatase \\
    \hline
    $A_{21}$ & active DC bound with phosphatase in the nucleus \\
    \hline
    $A_{22}$ & GSK3$\beta$ \\
    \hline
    $A_{23}$ & axin-APC complex \\
    \hline
    $A_{24}$ & APC \\
    \hline
    $A_{25}$ & $\beta$-catenin bound with DC (for proteasomal degradation) \\
    \hline
    $A_{26}$ & axin \\
    \hline
    $A_{27}$ & $\beta$-catenin-axin complex \\
    \hline
    $A_{28}$ & a complex considered as a single species (in \cite{FeinbergBook}): \\
    & ($A_{13}+A_{22}+A_{23}=A_{28}$)\\
    \hline
\end{tabular}
\end{table}

\section{Reaction networks of Wnt signaling models}\label{reaction:networks}

\subsection{Lee Model}
\label{app:Lee}

The following is the reaction network for the Lee model:

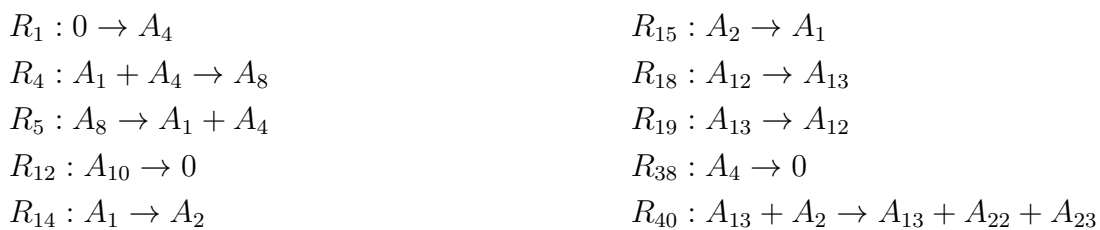
\allowdisplaybreaks
\begin{multicols}{2}
\noindent
\begin{align*}
& R_1: 0 \rightarrow A_4 \\
& R_4: A_1 + A_4 \rightarrow A_8 \\
& R_5: A_8 \rightarrow A_1  +  A_4 \\
& R_{12}: A_{10} \rightarrow 0 \\
& R_{14}: A_1 \rightarrow A_2 \\
& R_{15}: A_2 \rightarrow A_1 \\
& R_{18}: A_{12} \rightarrow A_{13} \\
& R_{19}: A_{13} \rightarrow A_{12} \\
& R_{38}: A_4 \rightarrow 0 \\
& R_{40}: A_{13} + A_2 \rightarrow A_{13} + A_{22} + A_{23} \\
& R_{41}: A_{22} + A_{23} \rightarrow A_2 \\
& R_{42}: A_2 \rightarrow A_{22} + A_{23} \\
& R_{43}: A_{24} + A_{26} \rightarrow A_{23} \\
& R_{44}: A_{23} \rightarrow A_{24} + A_{26} \\
& R_{45}: A_8 \rightarrow A_{25} \\
& R_{46}: A_{25} \rightarrow A_1 + A_{10} \\
& R_{47}: 0 \rightarrow A_{26} \\
& R_{48}: A_{26} \rightarrow 0 \\
& R_{49}: A_4 + A_6 \rightarrow A_7 \\
& R_{50}: A_7 \rightarrow A_4 + A_6 \\
& R_{51}: A_{24} + A_4 \rightarrow A_{27} \\
& R_{52}: A_{27} \rightarrow A_{24} + A_4
\end{align*}
\end{multicols}

\subsection{Feinberg Model}
\label{app:FAL}

The following is the reaction network for the Feinberg model:

\begin{multicols}{3}
\noindent
\begin{align*}
& R_1: 0 \rightarrow A_4 \\
& R_4: A_1 + A_4 \rightarrow A_8 \\
& R_5: A_8 \rightarrow A_1  +  A_4 \\
& R_{12}: A_{10} \rightarrow 0 \\
& R_{14}: A_1 \rightarrow A_2 \\
& R_{15}: A_2 \rightarrow A_1 \\
& R_{18}: A_{12} \rightarrow A_{13} \\
& R_{19}: A_{13} \rightarrow A_{12} \\
& R_{38}: A_4 \rightarrow 0 \\
& R_{43}: A_{24} + A_{26} \rightarrow A_{23} \\
& R_{44}: A_{23} \rightarrow A_{24} + A_{26} \\
& R_{45}: A_8 \rightarrow A_{25} \\
& R_{46}: A_{25} \rightarrow A_1 + A_{10} \\
& R_{47}: 0 \rightarrow A_{26} \\
& R_{48}: A_{26} \rightarrow 0 \\
& R_{49}: A_4 + A_6 \rightarrow A_7 \\
& R_{50}: A_7 \rightarrow A_4 + A_6 \\
& R_{51}: A_{24} + A_4 \rightarrow A_{27} \\
& R_{52}: A_{27} \rightarrow A_{24} + A_4 \\
& R_{53}: A_{13} + A_2 \rightarrow A_{28} \\
& R_{54}: A_2 \rightarrow A_{23} \\
& R_{55}: A_{23} \rightarrow A_2 \\
& R_{56}: A_{28} \rightarrow A_{13} + A_{23}
\end{align*}
\end{multicols}

\subsection{Schmitz Model}
\label{app:Schmitz}

The following is the reaction network for the  Schmitz model:

\begin{multicols}{3}
\noindent
\begin{align*}
& R_1: 0 \rightarrow A_4 \\
& R_2: A_4 \rightarrow A_5 \\
& R_3: A_5 \rightarrow A_4 \\
& R_4: A_1 + A_4 \rightarrow A_8 \\
& R_5: A_8 \rightarrow A_1  +  A_4 \\
& R_6: A_5 + A_3 \rightarrow A_9 \\
& R_7: A_9 \rightarrow A_5  +  A_3 \\
& R_8: A_6 + A_5 \rightarrow A_7 \\
& R_9: A_7 \rightarrow A_6  +  A_5 \\
& R_{10}: A_8 \rightarrow A_1 + A_{10} \\
& R_{11}: A_9 \rightarrow A_3 + A_{11} \\
& R_{12}: A_{10} \rightarrow 0 \\
& R_{13}: A_{11} \rightarrow 0 \\
& R_{14}: A_1 \rightarrow A_2 \\
& R_{15}: A_2 \rightarrow A_1 \\
& R_{16}: A_1 \rightarrow A_3 \\
& R_{17}: A_3 \rightarrow A_1
\end{align*}
\end{multicols}

\subsection{MacLean Model}
\label{app:MacLean}

The following is the reaction network for the MacLean model:

\begin{multicols}{3}
\noindent
\begin{align*}
& R_1: 0 \rightarrow A_4 \\
& R_2: A_4 \rightarrow A_5 \\
& R_3: A_5 \rightarrow A_4 \\
& R_4: A_1 + A_4 \rightarrow A_8 \\
& R_5: A_8 \rightarrow A_1  +  A_4 \\
& R_6: A_5 + A_3 \rightarrow A_9 \\
& R_7: A_9 \rightarrow A_5  +  A_3 \\
& R_8: A_6 + A_5 \rightarrow A_7 \\
& R_9: A_7 \rightarrow A_6  +  A_5 \\
& R_{18}: A_{12} \rightarrow A_{13} \\
& R_{19}: A_{13} \rightarrow A_{12} \\
& R_{20}: A_{13} \rightarrow A_{14} \\
& R_{21}: A_{14} \rightarrow A_{13} \\
& R_{22}: A_2 \rightarrow A_{15} \\
& R_{23}: A_{15} \rightarrow A_2 \\
& R_{24}: A_3  +  A_{14} \rightarrow A_{19} \\
& R_{25}: A_{19} \rightarrow A_3  +  A_{14} \\
& R_{26}: A_{19} \rightarrow A_{14}  +  A_{15} \\
& R_{27}: A_{15}  +  A_{17} \rightarrow A_{21} \\
& R_{28}: A_{21} \rightarrow A_{15}  +  A_{17} \\
& R_{29}: A_{21} \rightarrow A_3  +  A_{17} \\
& R_{30}: A_{13}  +  A_1 \rightarrow A_{18} \\
& R_{31}: A_{18} \rightarrow A_{13}  + A_1 \\
& R_{32}: A_{18} \rightarrow A_{13}  + A_2 \\
& R_{33}: A_2  +  A_{16} \rightarrow A_{20} \\
& R_{34}: A_{20} \rightarrow A_2  +  A_{16} \\
& R_{35}: A_{20} \rightarrow A_1  +  A_{16} \\
& R_{36}: A_8 \rightarrow A_1 \\
& R_{37}: A_9 \rightarrow A_3 \\
& R_{38}: A_4 \rightarrow 0 \\
& R_{39}: A_5 \rightarrow 0
\end{align*}
\end{multicols}

\end{document}